\documentclass[preprint,12pt]{elsarticle}

\usepackage[utf8]{inputenc}
\usepackage{cases}
\usepackage{amsthm}
\usepackage{amsfonts}
\usepackage{multicol}
\usepackage{graphicx}
\usepackage{epsfig}
\usepackage{pdfpages}
\usepackage{amssymb}

\usepackage{grffile}        
\usepackage{caption}
\usepackage{subcaption}

\DeclareGraphicsExtensions{.jpg}

\journal{Communications in Nonlinear Science and Numerical Simulation}

\begin{document}

\newtheorem{conjectura}{Conjecture}
\newtheorem{definition}{Definition}
\newtheorem{lema}{Lema}
\newtheorem{theorem}{Theorem}
\newtheorem{proposition}{Proposition}

\begin{frontmatter}

\title{On the mKdV-Liouville hierarchy and its self-similarity reduction}

\author{Danilo V. Ruy}

\address{Instituto de F\'isica Te\'orica-UNESP, Rua Dr Bento Teobaldo Ferraz 271, Bloco II, 01140-070, S\~ao Paulo, Brazil}
\ead{daniloruy@ift.unesp.br}

\begin{abstract}
Integrable mixed models have been used as a generalization of traditional integrable models. However, a map from a traditional integrable model to a mixed integrable model is not well understood yet. Here, it is studied the relation between the mKdV-Liouville hierarchy and the mKdV hierarchy by employing an extended version of the modified truncation approach. This paper shows some solutions for the mKdV-Liouville hierarchy constructed from the soliton solutions of the mKdV hierarchy. The last section deals with the possibility of define new transcendental functions from the self-similarity reduction of the mKdV-Liouville hierarchy.

\end{abstract}


\end{frontmatter}


\section{Introduction}

The mKdV equation has been widely studied in the last decades and possess some well know solutions \cite{Lamb,Chen,Mosavian,Franca}. A generalization of the mKdV equation which combines the mKdV and the sinh-Gordon equations was proposed in \cite{konno1} as a mixed integrable model. Later, this mixed model showed to be suitable for describing few-optical-cycles pulse in transparent media \cite{Leblon}. The generalized mKdV-sinh-Gordon hierarchy \cite{Ruy} is an mixed model which include the mixed mKdV-sinh-Gordon ($\delta(t)=-\beta(t)$) and the mixed mKdV-Liouville hierarchies ($\delta(t)=0$), namely 
\begin{equation}\label{mKdV-SG}
    \epsilon_0(t){\partial\over\partial x}\biggl({\partial\over\partial x}+y_x\biggr){\cal L}_{n}\biggl[x;y_{xx}-{1\over2}y_x^2\biggr]+y_{xt}+\beta(t)e^y+\delta(t)e^{-y}=0,
\end{equation}
where ${\cal L}_{n}[x;u]$ is the Lenard recurrence relation, i.e.
\[
{\partial\over\partial x}{\cal L}_{n+1}[x;u]=\biggl({\partial^3\over\partial x^3}+2u{\partial\over\partial x}+u_x\biggr){\cal L}_{n}[x;u],  \hspace{1 cm}   {\cal L}_{0}[x;u]=1,
\]
such that the first equation of the hierarchy is
\begin{equation}\label{mKdV-SG eq1}
 \epsilon_0(t)\biggl(y_{xxxx}-{3\over2}y_x^2y_{xx}\biggr)+y_{xt}+\beta(t)e^y+\delta(t)e^{-y}=0
\end{equation}


It was shown in \cite{Ruy} that the self-similarity reduction of (\ref{mKdV-SG eq1}) yields Kudryashov's equation \cite{kudryashov}. Such equation passes the necessary condition for absence of movable branches points, called Painlev\'e property, and it can reduce to two Painlev\'e equation by appropriate choices of the parameters. The Painlev\'e equations are second-order nonlinear ODEs (ordinary differential equations) which define new transcendental functions \cite{Umemura} and it has motivated many studies in higher order ODEs \cite{Ruy,kudryashov,cosgrove1,cosgrove2,gordoa,k1,k2,mugan1,mugan2,sakka,sakka2,sakka3,k3}. Due the connection with two Painlev\'e equations, Kudryashov's equation was conjectured as a possibility of defining a new transcendental function \cite{kudryashov,kudryashov-Sinelshchikov}. The representation of the solutions of (\ref{mKdV-SG}) and its self-similarity reduction in terms of a simpler hierarchy is not well understood yet. In this paper, we study this relation in the particular case of the mKdV-Liouville hierarchy.

The modified truncation approach was introduced in \cite{Clarkson} as a technique to find auto-B\"acklund transformations for an hierarchy of ODEs. Here, we employ an extended version of this approach in order to obtain one transformation from soliton solutions for the mKdV hierarchy to the mKdV-Liouville hierarchy. Besides, we obtain a transformation from the PII hierarchy to a particular case of the Kudryashov's hierarchy.





\section{Transformations for the mKdV-Liouville hierarchy}

Let us consider the hierarchy (\ref{mKdV-SG}) with $\delta(t)=0$ and $\epsilon_0(t)\neq 0$.  Due to theorem \ref{escala} of appendix, we can choose $\epsilon_0(t)=1$ without lost of generality. So we consider the mKdV-Liouville hierarchy as
\begin{equation}\label{mKdV-Liouville}
{\partial\over\partial x}\biggl({\partial\over\partial x}+y_x\biggr){\cal L}_{n}\biggl[x;y_{xx}-{1\over2}y_x^2\biggr]+y_{xt}+\beta(t)e^y=0.
\end{equation}



In order to transform ${\cal L}_{n}\biggl[x;y_{xx}-{1\over2}y_x^2\biggr]$ in a M\"{o}bius invariant element, we do the transformation
\begin{equation}\label{transf sigma}
y=\ln[g(t)\sigma_x],
\end{equation}
where $g(t)$ is a arbitrary function.  Thus, the hierarchy (\ref{mKdV-Liouville}) becomes
\begin{equation}\label{hierarchy sigma}
 {\partial\over\partial x}\biggl({\partial\over\partial x}+{\sigma_{xx}\over\sigma_x}\biggr){\cal L}_{n}[x;S(\sigma)]+{d\over dt}\biggl({\sigma_{xx}\over\sigma_x}\biggr)+\beta(t)g(t)\sigma_x=0
\end{equation}
where $S(\sigma)$ is the Schwarzian derivative, i.e.
\[
S(\sigma)={d\over dx}\biggl({\sigma_{xx}\over\sigma_x}\biggr)-{1\over2}\biggl({\sigma_{xx}\over\sigma_x}\biggr)^2.
\]

By doing the M\"obius transformation $\sigma=-1/\phi$ and defining  $\tilde{y}=\ln[g(t)\phi_x]$, we obtain the relation
\begin{equation}\label{y tilde}
y=\ln\biggl[g(t){{\phi_x\over\phi^2}}\biggr]=\ln[g(t)\phi_x]-2\ln\phi \equiv \tilde{y}-2\ln\phi
\end{equation}
and the hierarchy (\ref{hierarchy sigma}) is transformed into
\begin{equation}
 {\partial\over\partial x}\biggl({\partial\over\partial x}+{\phi_{xx}\over\phi_x}-2{\phi_x\over\phi}\biggr){\cal L}_{n}[x;S(\phi)]+{d\over dt}\biggl({\phi_{xx}\over\phi_x}-2{\phi_x\over\phi}\biggr)+\beta(t)g(t){\phi_x\over\phi^2}=0,
\end{equation}
which it can be rewrite as
\begin{equation}\label{y and phi}
 {\partial\over\partial x}\biggl({\partial\over\partial x}+\tilde{y}_x-2{\phi_x\over\phi}\biggr){\cal L}_{n}\biggl[x;\tilde{y}_{xx}-{1\over2}\tilde{y}_x^2\biggr]+{d\over dt}\biggl(\tilde{y}_x-2{\phi_x\over\phi}\biggr)+\beta(t)g(t){\phi_x\over\phi^2}=0
\end{equation}

Two cases are considered  in the following subsections by using equation (\ref{y and phi}), i. e. the partial differential equation (\ref{mKdV-Liouville}) and a self-similarity reduction of it.



\subsection{From a soliton solution of the mKdV hierarchy to a mKdV-Liouville solution}

Observe that $g(t)$ only appears multiplying $\beta(t)$ in (\ref{y and phi}). As we can redefine $\beta(t)$ using theorem \ref{escala} of appendix, we can assume $g(t)=1$ without lost of generality.

If $\tilde{y}_x=2v$, such that $v$ satisfies the mKdV hierarchy, i.e.
\begin{equation}\label{mKdV hierarchy}
 {\partial\over\partial x}\biggl({\partial\over\partial x}+2v\biggr){\cal L}_{n}\biggl[x;2(v_x-v^2)\biggr]+2v_t=0,
\end{equation}
then expression (\ref{y and phi}) yields the following condition on $\phi$
\begin{equation}
 {\partial\over\partial x}\biggl(-2{\phi_x\over\phi}{\cal L}_{n}\biggl[x;\tilde{y}_{xx}-{1\over2}\tilde{y}_x^2\biggr]\biggr)+{d\over dt}\biggl(-2{\phi_x\over\phi}\biggr)+\beta(t){\phi_x\over\phi^2}=0
\end{equation}

By integrating the above equation in $x$, we have
\begin{equation}\label{aux}
2\phi_x{\cal L}_{n}\biggl[x;\tilde{y}_{xx}-{1\over2}\tilde{y}_x^2\biggr]+2\phi_t+\beta(t)+\alpha(t)\phi=0,
\end{equation}
where $\alpha(t)$ is an arbitrary function. Condition (\ref{aux}) is similar to one found in \cite{Weiss 2} for the mKdV equation, although it was not derived any explicit solutions there.

In order to check that the equation (\ref{aux}) is compatible with the assumption (\ref{mKdV hierarchy}), let us isolate ${\cal L}_{n}\biggl[x;\tilde{y}_{xx}-{1\over2}\tilde{y}_x^2\biggr]$, i. e.
\begin{equation}
{\cal L}_{n}\biggl[x;\tilde{y}_{xx}-{1\over2}\tilde{y}_x^2\biggr]=-{\phi_t\over\phi_x}-{(\beta(t)+\alpha(t)\phi)\over 2\phi_x}
\end{equation}

The above expression together with the definition $\tilde{y}=\ln[\phi_x]$ yields
\begin{equation}\label{aux 1}
{\partial\over\partial x}{\cal L}_{n}\biggl[x;\tilde{y}_{xx}-{1\over2}\tilde{y}_x^2\biggr]=-{\phi_{xt}\over\phi_x}+{\phi_{t}\phi_{xx}\over\phi_x^2}+{(\beta(t)+\alpha(t)\phi)\phi_{xx}\over 2\phi_x^2}-{\alpha(t)\over2}
\end{equation}
\begin{equation}\label{aux 2}
\tilde{y}_x{\cal L}_{n}\biggl[x;\tilde{y}_{xx}-{1\over2}\tilde{y}_x^2\biggr]=-{\phi_{t}\phi_{xx}\over\phi_x^2}-{(\beta(t)+\alpha(t)\phi)\phi_{xx}\over 2\phi_x^2}
\end{equation}

By summing (\ref{aux 1}) and (\ref{aux 2}), we have

\begin{equation}\label{int y tilde}
\biggl({\partial\over\partial x}+\tilde{y}_x\biggr){\cal L}_{n}\biggl[x;\tilde{y}_{xx}-{1\over2}\tilde{y}_x^2\biggr]+\tilde{y}_t+{\alpha(t)\over2}=0,
\end{equation}
which gives (\ref{mKdV hierarchy}) through a derivation and the transformation $\tilde{y}_x=2v$. Hence $\tilde{y}$ is given in terms of $v$ by
\begin{equation}\label{sol v}
\tilde{y}=2\int^x v(x',t) dx'+\Gamma(t)
\end{equation}
where $\Gamma(t)$ is an arbitrary function. If we consider a soliton solution for $v$, we must choose $\alpha(t)=-2\Gamma'(t)$. This can be verified by  substituting  (\ref{sol v}) in (\ref{mKdV hierarchy}). Thus, we can rewrite equation (\ref{aux}) as
%
%
\begin{equation}\label{general equation}
 {\partial\over\partial t}\biggl(e^{-\Gamma(t)}\phi\biggr) +{\beta(t)e^{-\Gamma(t)}\over2}=-e^{\tilde{y}-\Gamma(t)}{\cal L}_{n}\biggl[x;\tilde{y}_{xx}-{1\over2}\tilde{y}_x^2\biggr]
\end{equation}

As the right hand side of the equation (\ref{general equation}) is expressed in terms of known elements, we can integrate it in order to determine $\phi$, i.e.
\begin{equation}\label{phi}
\phi =-e^{\Gamma(t)}\biggl(\int^t_{-\infty} e^{\tilde{y}-\Gamma(t')}{\cal L}_{n}\biggl[x;\tilde{y}_{xx}-{1\over2}\tilde{y}_x^2\biggr]dt'+\int^t_{-\infty}{\beta(t')e^{-\Gamma(t')}\over2}dt' +\Phi_0(x)\biggr) 
\end{equation}
where $\Phi_0(x)$ must be determined for (\ref{phi}) to be consistent with the definition $\tilde{y}=\ln\phi_x$. Let us derive (\ref{phi}) in $x$ and use (\ref{int y tilde}) such that
\begin{equation}\label{phi x}
{\partial\phi\over\partial x}=-e^{\Gamma(t)}\biggl(\int^t_{-\infty} e^{\tilde{y}-\Gamma(t')}\biggl({\partial\over\partial x}+\tilde{y}_x\biggr){\cal L}_{n}\biggl[x;\tilde{y}_{xx}-{1\over2}\tilde{y}_x^2\biggr]dt'+\Phi_0'(x)\biggr) = e^{\tilde{y}}-e^{\Gamma(t)}\biggl(\lim_{t'\to-\infty}e^{\tilde{y}(x,t')-\Gamma(t')}+\Phi_0'(x)\biggr).
\end{equation}
By using the definition of $\tilde{y}$, the above expression yields the condition
\[
\Phi_0(x)=-\int^x\lim_{t'\to-\infty}e^{\tilde{y}(x',t')-\Gamma(t')}dx'+c_1,    \hspace{1 cm} c_1 \equiv\textnormal{constant}
\]

Observe that a soliton solution from the mKdV hierarchy yields $\Phi_0(x)=-x+c_1$. Let us see three examples that work for the whole mKdV-Liouville hierarchy:

{\it  Example 1)} The vacuum solution of the mKdV hierarchy, i. e. $v=0$ yields $\tilde{y}=\Gamma(t)$ and
\[
\phi=e^{\Gamma(t)}\biggl(-\int^t{\beta(t')e^{-\Gamma(t')}\over2}dt' +x-c_1  \biggr)
\]
\begin{equation}  \label{ex 1}
y=-2\ln\biggl(-\int^t{\beta(t')e^{-\Gamma(t')}\over2}dt' +x-c_1 \biggr)-\Gamma(t)
\end{equation}

{\it  Example 2)} The 1-soliton  solution of the mKdV hierarchy, i. e.
\[
v= {\partial\over\partial x}\ln\biggl({2-e^{\eta}\over2+e^{\eta}}\biggr)  \hspace{1.5 cm}  \eta=k x-k^{2n+1} t ,
\]
yields
\[
\phi=e^{\Gamma(t)}\biggl(-{4e^{\eta}\over k(2+e^{\eta})}-\int^t{\beta(t')e^{-\Gamma(t')}\over2}dt' +x-c_1  \biggr)
\]
\begin{equation}  \label{ex 2}
y=2\ln\biggl({2-e^{\eta}\over2+e^{\eta}}\biggr)-2\ln\biggl(-{4e^{\eta}\over k(2+e^{\eta})}-\int^t{\beta(t')e^{-\Gamma(t')}\over2}dt' +x-c_1 \biggr)-\Gamma(t)
\end{equation}

{\it  Example 3)} The 2-soliton solution of the mKdV hierarchy, i. e
\[
v= {\partial\over\partial x}\ln\pmatrix{{4+2(e^{\eta_1}+e^{\eta_2})+\biggl({k_1-k_2\over k_1+k_2}\biggr)^2 e^{\eta_1+\eta_2}\over 4-2(e^{\eta_1}+e^{\eta_2})+\biggl({k_1-k_2\over k_1+k_2}\biggr)^2 e^{\eta_1+\eta_2}} } ,       \hspace{1.5 cm}     \eta_j=k_j x-k_j^{2n+1} t    ,  \hspace{1 cm}    j=1,2
\]
yields
\begin{eqnarray*}
\phi &=& e^{\Gamma(t)}\Biggl(-{4(k_1+k_2)\over k_1k_2}\biggl({k_1^2 e^{\eta_2}(e^{\eta_1}-2)-2k_1k_2(e^{\eta_1}+e^{\eta_2}+e^{\eta_1+\eta_2})+k_2^2e^{\eta_1}(e^{\eta_2}-2) \over  k_1^2(e^{\eta_1}-2)(e^{\eta_2}-2)-2k_1k_2(2e^{\eta_1}+2e^{\eta_2}+e^{\eta_1+\eta_2}-4)+k_2^2(e^{\eta_1}-2)(e^{\eta_2}-2) }\biggr)  \\
 &-& \int^t{\beta(t')e^{-\Gamma(t')}\over2}dt' +x-c_1  \Biggr)
\end{eqnarray*}
\begin{eqnarray}
y &=& 2\ln\pmatrix{{4+2(e^{\eta_1}+e^{\eta_2})+\biggl({k_1-k_2\over k_1+k_2}\biggr)^2 e^{\eta_1+\eta_2}\over 4-2(e^{\eta_1}+e^{\eta_2})+\biggl({k_1-k_2\over k_1+k_2}\biggr)^2 e^{\eta_1+\eta_2}} }   \nonumber \\
 & - & 2\ln\Biggl(-{4(k_1+k_2)\over k_1k_2}\biggl({k_1^2 e^{\eta_2}(e^{\eta_1}-2)-2k_1k_2(e^{\eta_1}+e^{\eta_2}+e^{\eta_1+\eta_2})+k_2^2e^{\eta_1}(e^{\eta_2}-2) \over  k_1^2(e^{\eta_1}-2)(e^{\eta_2}-2)-2k_1k_2(2e^{\eta_1}+2e^{\eta_2}+e^{\eta_1+\eta_2}-4)+k_2^2(e^{\eta_1}-2)(e^{\eta_2}-2) }\biggr)  \nonumber \\
 & - & \int^t{\beta(t')e^{-\Gamma(t')}\over2}dt' +x-c_1  \Biggr)-\Gamma(t)   \label{ex 3}
\end{eqnarray}

If we choose $\beta(t)=0$, the results of this section represent a map from the mKdV hierarchy into itself. By defining the field $\tilde{v}=y_x/2$, we have solutions for the mKdV hierarchy which appears to be new to the author knowledge.  Below, we show the solutions for the mKdV equation which can be obtained from the examples we have just used:

\begin{eqnarray*}
\textnormal{Example 1)} & &\tilde{v}={1\over c_1-x}  \\
&& \\
\textnormal{Example 2)} & & \tilde{v}={k[4+4e^{\eta}(k(x-c_1)-2)-e^{2\eta}]\over 4k(c_1-x)+8e^{\eta}+e^{2\eta}(k(x-c_1)-4)} \\
&& \\
\textnormal{Example 3)} & & \tilde{v}=\biggl[-4+4e^{\eta_1}(k_1(x-c_1)-2)+4e^{\eta_2}(k_2(x-c_1)-2)+e^{2\eta_1}+e^{2\eta_2}  \\
&&                                                        +{8(k_1^4-k_1^2k_2^2+k_2^4)\over k_1k_2(k_1+k_2)^2}e^{\eta_1+\eta_2}-{(k_1^2-k_2^2)^2\over k_1k_2(k_1+k_2)^2}[k_1e^{\eta_1+2\eta_2}(k_1k_2(x-c_1)-4k_1-2k_2)  \\
&&                                                        +k_2e^{2\eta_1+\eta_2}(k_1k_2(x-c_1)-2k_1-4k_2)]-{(k_1-k_2)^4\over 4(k_1+k_2)^4}e^{2(\eta_1+\eta_2)}              \biggr]\biggl/\biggl[4(x-c_1) \\
&&                                                        +{8e^{\eta_1}\over k_1}+{8e^{\eta_2}\over k_2}+e^{2\eta_1}\biggl(c_1-x+{4\over k_1}\biggr)+e^{2\eta_2}\biggl(c_1-x+{4\over k_2}\biggr) \\
&&                                                        +{8k_1k_2(c_1-x)+2(k_1+k_2)\over (k_1+k_2)^2}e^{\eta_1+\eta_2}-{2(k_1-k_2)^2\over k_1(k_1+k_2)^2}e^{\eta_1+2\eta_2}  \\
&&                                                        -{2(k_1-k_2)^2\over k_2(k_1+k_2)^2}e^{2\eta_1+\eta_2}+{(k_1-k_2)^4(k_1k_2(x-c_1)-4(k_1+k_2))\over 4k_1k_2(k_1+k_2)^4}e^{2(\eta_1+\eta_2)}\biggr]
\end{eqnarray*}

\begin{figure}
        \centering
        \begin{subfigure}[hbt]{0.3\textwidth}
                \includegraphics[width=\textwidth]{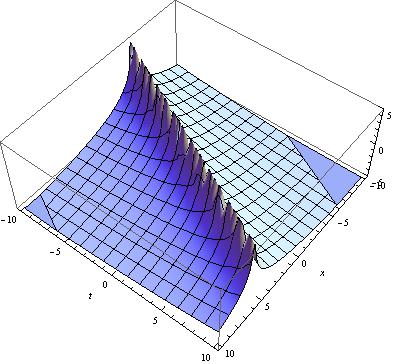}
                \caption{Example 1}
                \label{fig vacuo}
        \end{subfigure}%
        \begin{subfigure}[hbt]{0.3\textwidth}
                \includegraphics[width=\textwidth]{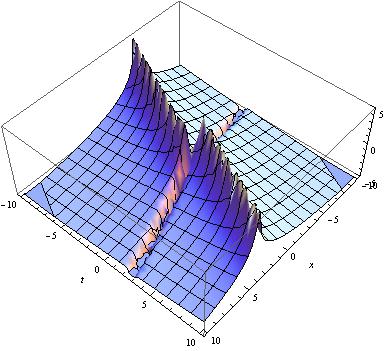}
                \caption{Example 2}
                \label{fig 1 sol}
        \end{subfigure}
        \begin{subfigure}[hbt]{0.3\textwidth}
                \includegraphics[width=\textwidth]{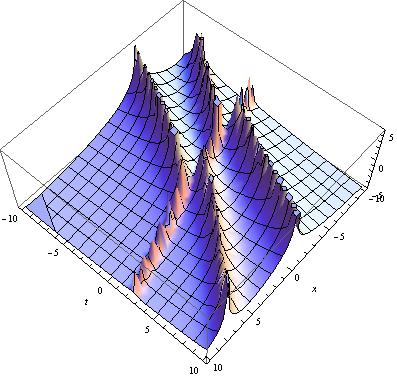}
                \caption{Example 3}
                \label{fig 2 sol}
        \end{subfigure}
        \caption{It was showed we can obtain solutions for the mKdV-Liouville hierarchy from the solutions of the mKdV hierarchy. In these figures, we illustrate some solutions for the first equation of the mKdV-Liouville hierarchy ($n=1$) with $\beta(t)=1$. Figure \ref{fig vacuo} shows solution (\ref{ex 1}) with $\Gamma(t)=0$ and $c_1=1$. Figure \ref{fig 1 sol} shows solution (\ref{ex 2}) with $\Gamma(t)=0$, $c_1=1$ and $k=-2$. Figure \ref{fig 2 sol} shows solution (\ref{ex 3}) with $\Gamma(t)=0$, $c_1=1$, $k_1=-1$ and $k_2=-2$.    }
\end{figure}


\subsection{Self-similarity reduction}

Let us consider the expression (\ref{mKdV hierarchy}) again. By theorem \ref{escala} of the appendix, we always can choose
\[
\beta(t)=\beta_0[(2n+1)t]^{-(2n+2)/(2n+1)}.
\]

Thus, by assuming the self-similarity reduction
\[
z=x[(2n+1)t]^{-1/(2n+1)} \hspace{0.7 cm} y(x,t)=y(z), \hspace{0.7 cm}   {\cal L}_{k}\biggl[x;y_{xx}-{1\over2}y_x^2\biggr]=[(2n+1)t]^{-2k/(2n+1)}{\cal L}_{k}\biggl[z;y_{zz}-{1\over2}y_z^2\biggr],
\]
the mKdV-Liouville hierarchy reduce to
\begin{equation}\label{mkdv-L in z}
{\partial\over\partial z}\biggl({\partial\over\partial z}+y_z\biggr){\cal L}_{n}\biggl[z;y_{zz}-{1\over2}y_z^2\biggr]-(zy_z)_z+\beta_0 e^{y}=0
\end{equation}

Hierarchy (\ref{mkdv-L in z}) is a particular case of the hierarchy proposed in \cite{kudryashov}. 
%
If we choose $g(t)=[(2n+1)t]^{1/(2n+1)}$, the definition (\ref{y tilde}) reduce to
\[
\tilde{y}(x,t)=\tilde{y}(z) , \hspace{1 cm}  \phi(x,t)=\phi(z)
\]

Thus, the self-similarity reduction of  (\ref{y and phi}) yields
\begin{equation}
 {\partial\over\partial z}\biggl({\partial\over\partial z}+\tilde{y}_z-2{\phi_z\over\phi}\biggr){\cal L}_{n}\biggl[z;\tilde{y}_{zz}-{1\over2}\tilde{y}_z^2\biggr]-{d\over dz}\biggl(z\tilde{y}_z-2z{\phi_z\over\phi}\biggr)+\beta_0{\phi_z\over\phi^2}=0
\end{equation}
which can be integrated and the integration constant defined as $2(\alpha-1)$, i. e.
\begin{equation}\label{eq z}
\biggl({\partial\over\partial z}+\tilde{y}_z-2{\phi_z\over\phi}\biggr){\cal L}_{n}\biggl[z;\tilde{y}_{zz}-{1\over2}\tilde{y}_z^2\biggr]-z\tilde{y}_z+2z{\phi_z\over\phi}-{\beta_0\over\phi}+2(\alpha-1)=0
\end{equation}

Let us assume that $\tilde{y}$ satisfies
\begin{equation}\label{PII in y}
\biggl({\partial\over\partial z}+\tilde{y}_z\biggr){\cal L}_{n}\biggl[z;\tilde{y}_{zz}-{1\over2}\tilde{y}_z^2\biggr]-z\tilde{y}_z-2\alpha=0,
\end{equation}
such that (\ref{eq z}) is simplified to
\begin{equation}\label{phi z}
-2{\phi_z\over\phi}{\cal L}_{n}\biggl[z;\tilde{y}_{zz}-{1\over2}\tilde{y}_z^2\biggr]+2z{\phi_z\over\phi}-{\beta_0\over\phi}+2(2\alpha-1)=0
\end{equation}

In order to check the compatibility between (\ref{PII in y}) and (\ref{phi z}), let us isolate ${\cal L}_{n}\biggl[z;\tilde{y}_{zz}-{1\over2}\tilde{y}_z^2\biggr]$, i.e.
\begin{equation}
{\cal L}_{n}\biggl[z;\tilde{y}_{zz}-{1\over2}\tilde{y}_z^2\biggr]=z-{\beta_0\over 2\phi_z}+(2\alpha-1){\phi \over\phi_z}
\end{equation}
such that
\[
{\partial\over\partial z}{\cal L}_{n}\biggl[z;\tilde{y}_{zz}-{1\over2}\tilde{y}_z^2\biggr]={\beta_0\phi_{zz}\over 2\phi_z^2}+2\alpha-(2\alpha-1){\phi\phi_{zz}\over \phi_z^2}
\]
\[
\tilde{y}_z{\cal L}_{n}\biggl[z;\tilde{y}_{zz}-{1\over2}\tilde{y}_z^2\biggr]=z\tilde{y}_z-{\beta_0\phi_{zz}\over 2\phi_z^2}+(2\alpha-1){\phi\phi_{zz}\over \phi_z^2}
\]

It is trivial to check that  (\ref{PII in y}) is true with the above expressions. The transformation
\[
\tilde{y}=2\int^zv(z')dz'+C
\]
maps the hierarchy  (\ref{PII in y}) into the the PII hierarchy, i. e.
\begin{equation}
\biggl({\partial\over\partial z}+2v\biggr){\cal L}_{n}\biggl[z;2(v_z-v^2)\biggr]-2zv-2\alpha=0,
\end{equation}
which has the PII equation as the first equation,  namely
\[
v_{zz}=2v^3+zv+\alpha.
\]

From (\ref{phi z}), we have
\[
\phi={e^{\tilde{y}}\biggl({\cal L}_{n}\biggl[z;\tilde{y}_{zz}-{1\over2}\tilde{y}_z^2\biggr]-z\biggr)+\beta_0/2\over (2\alpha-1)},
\]
provided that $\alpha\neq 1/2$. Therefore, the solution for (\ref{mkdv-L in z}) is related with the solution of the PII hierarchy by the transformation
\begin{equation}\label{sol z}
y=2\int^zv(z')dz'+C+2\ln(2\alpha-1)-2\ln\Biggl[e^{\tilde{y}}\biggl({\cal L}_{n}\biggl[z;2(v_z-v^2)\biggr]-z\biggr)+\beta_0/2\Biggr]
\end{equation}

Observe that the first equation of the hierarchy (\ref{mkdv-L in z}) is a forth order equation, i.e.
\begin{equation}\label{kudryashov eq}
y_{zzzz}-{3\over2}y_z^2y_{zz}-zy_{zz}-y_z+\beta_0 e^y=0,
\end{equation}
which is a particular case of the equation proposed in \cite{kudryashov} as a possibility to define a new transcendental function.

The general solution of the PII equation define a transcendental solution with two arbitrary constants plus the parameter $\alpha$. The solution (\ref{sol z}) for $n=1$ has the arbitrary constants $C$, $\alpha$ and two constants of integration from the general solution of the PII equation. Thus, solution (\ref{sol z}) represent the general solution for the equation (\ref{kudryashov eq}), which is a particular case of the Kudryashov's equation.


As examples,  we show some solutions for equation (\ref{kudryashov eq}) based on rational solutions for the PII equation in table \ref{rational}.

\begin{table}[htp]\caption{Solutions for equation (\ref{kudryashov eq}) based on rational solutions from PII equation \label{rational}}
\begin{center}
\begin{tabular}{|l|l|l|}
  \hline
  $\alpha=-2$  &   $v={2(z^3-2)\over z(z^3+4)}$  &  $y=C-2\ln\biggl({2e^C(z^6+20z^3-80)-\beta_0 z\over 10(z^3+4)}\biggr)$ \\
&  &  \\
  \hline
  $\alpha=-1$  &   $v={1\over z}$  &  $y=C-2\ln\biggl({2e^C(z^3+4)-\beta\over 6z}\biggr)$ \\
&  &  \\
  \hline
  $\alpha=0$  &   $v=0$  &  $y=C-2\ln\biggl(e^C z-{\beta_0\over 2}\biggr)$ \\
&  &  \\
  \hline
  $\alpha=1$  &   $v=-{1\over z}$  &  $y=C-2\ln\biggl({\beta_0 z\over 2}-e^C\biggr)$ \\
&  &  \\
  \hline
  $\alpha=2$  &   $v=-{2(z^3-2)\over z(z^3+4)}$  &  $y=C-2\ln\biggl({\beta_0 (z^3+4)-2e^C\over 6z}\biggr)$ \\
&  &  \\
  \hline
\end{tabular}
\end{center}
\end{table}

\section{Conclusion}

In this paper, it was shown the relation between the mKdV and mKdV-Liouville hierarchies by using an extension of the modified truncation approach. Some solutions for the mKdV-Liouville was presented, such that a particular case of these solutions yields new solutions for the mKdV hierarchy. Also, it was showed the general solution for a particular case of the Kudryashov's equation in terms of the second Painlev\'e transcendent.



\section*{Acknowledgements}
I am thankful to J. F. Gomes and A. H. Zimerman for discussions. The author thanks  FAPESP (2010/18110-9) for financial support.


\appendix

\section{Properties of the generalized mKdV-sinh-Gordon hierarchy}

Consider the generalized mKdV-sinh-Gordon hierarchy
\begin{equation}\label{mKdV-SG Ap}
E_n(y;\epsilon_0(t),\beta(t),\delta(t)):  \hspace{0.5 cm}    \epsilon_0(t){\partial\over\partial x}\biggl({\partial\over\partial x}+y_x\biggr){\cal L}_{n}\biggl[x;y_{xx}-{1\over2}y_x^2\biggr]+y_{xt}+\beta(t)e^y+\delta(t)e^{-y}=0
\end{equation}
where ${\cal L}_{n}[x;u]$ is the Lenard recurrence relation, i.e.
\[
{\partial\over\partial x}{\cal L}_{n+1}[x;u]=\biggl({\partial^3\over\partial x^3}+2u{\partial\over\partial x}+u_x\biggr){\cal L}_{n}[x;u],  \hspace{1 cm}   {\cal L}_{0}[x;u]=1.
\]

In this appendix we show that the generalized mKdV-sinh-Gordon hierarchy can be reduced in two simpler cases, namely, the mKdV-sinh-Gordon hierarchy or the mKdV-Liouville hierarchy. In order to show this, let us divide ${\cal L}_{n}[x;y_{xx}-{1\over2}y_x^2]$ in two parts as
\begin{equation}\label{odd-even}
{\cal L}_{n}\biggl[x;y_{xx}-{1\over2}y_x^2\biggr]={\cal L}_{n}^{(e)}\biggl[x;y_{xx}-{1\over2}y_x^2\biggr]+{\cal L}_{n}^{(o)}\biggl[x;y_{xx}-{1\over2}y_x^2\biggr].
\end{equation}
where we define ${\cal L}_{n}^{(e)}[x;y_{xx}-{1\over2}y_x^2]$ and ${\cal L}_{n}^{(o)}[x;y_{xx}-{1\over2}y_x^2]$ as the parts of ${\cal L}_{n}[x;y_{xx}-{1\over2}y_x^2]$ with even and odd dimension of the field respectively. For example
\[
{\cal L}_{1}\biggl[x;y_{xx}-{1\over2}y_x^2\biggr]=y_{xx}-{1\over2}y_x^2  \hspace{0.5 cm}  \Rightarrow   \hspace{0.5 cm} {\cal L}_{1}^{(e)}\biggl[x;y_{xx}-{1\over2}y_x^2\biggr]=-{1\over2}y_x^2   ,   \hspace{0.5 cm} {\cal L}_{1}^{(o)}\biggl[x;y_{xx}-{1\over2}y_x^2\biggr]=y_{xx} .
\]

From the definition of equation (\ref{odd-even}), observe that
\[
{\cal L}_{n}^{(e)}\biggl[x;y_{xx}-{1\over2}y_x^2\biggr]={\cal L}_{n}^{(e)}\biggl[x;-y_{xx}-{1\over2}y_x^2\biggr]
\]
\[
{\cal L}_{n}^{(o)}\biggl[x;y_{xx}-{1\over2}y_x^2\biggr]=-{\cal L}_{n}^{(o)}\biggl[x;-y_{xx}-{1\over2}y_x^2\biggr]
\]

Hence, the Lenard recurrence relation is equivalent to the following system
\begin{equation}\label{Lenard e}
{\partial\over\partial x}{\cal L}_{n+1}^{(e)}\biggl[x;y_{xx}-{1\over2}y_x^2\biggr]=\biggl({\partial^3\over\partial x^3}-y_x^2{\partial\over\partial x}-{1\over2}(y_x^2)_x\biggr){\cal L}_{n}^{(e)}\biggl[x;y_{xx}-{1\over2}y_x^2\biggr]+\biggl(2y_{xx}{\partial\over\partial x}+y_{xxx}\biggr){\cal L}_{n}^{(o)}\biggl[x;y_{xx}-{1\over2}y_x^2\biggr]
\end{equation}

\begin{equation}\label{Lenard o}
{\partial\over\partial x}{\cal L}_{n+1}^{(o)}\biggl[x;y_{xx}-{1\over2}y_x^2\biggr]=\biggl({\partial^3\over\partial x^3}-y_x^2{\partial\over\partial x}-{1\over2}(y_x^2)_x\biggr){\cal L}_{n}^{(o)}\biggl[x;y_{xx}-{1\over2}y_x^2\biggr]+\biggl(2y_{xx}{\partial\over\partial x}+y_{xxx}\biggr){\cal L}_{n}^{(e)}\biggl[x;y_{xx}-{1\over2}y_x^2\biggr]
\end{equation}

With the above properties, we can proof the following auto-B\"{a}cklund transformation.

\begin{theorem}\label{invertion}
Let $y=y(x,t)$ be the solution of $E_n(y;\epsilon_0(t),\beta(t),\delta(t))$, then $\tilde{y}=-y$ is the solution of $E_n(\tilde{y};\epsilon_0(t),-\delta(t),-\beta(t))$.
\end{theorem}

\begin{proof}
Consider the hierarchy
\begin{equation}
E_n(\tilde{y};\tilde{\epsilon}_0(t),\tilde{\beta}(t),\tilde{\delta}(t)):  \hspace{0.5 cm}    \tilde{\epsilon}_0(t){\partial\over\partial x}\biggl({\partial\over\partial x}+\tilde{y}_x\biggr){\cal L}_{n}\biggl[x;\tilde{y}_{xx}-{1\over2}\tilde{y}_x^2\biggr]+\tilde{y}_{xt}+\tilde{\beta}(t)e^{\tilde{y}}+\tilde{\delta}(t)e^{-\tilde{y}}=0
\end{equation}

By the transformation $\tilde{y}=-y$, we have
\[
\tilde{\epsilon}_0(t){\partial\over\partial x}\biggl({\partial\over\partial x}-y_x\biggr){\cal L}_{n}\biggl[x;-y_{xx}-{1\over2}y_x^2\biggr]-y_{xt}+\tilde{\beta}(t)e^{-y}+\tilde{\delta}(t)e^{y}=0
\]

Let us assume $\tilde{\epsilon}_0(t)=\epsilon_0(t)\neq 0$, $\tilde{\beta}(t)=-\delta(t)$ and $\tilde{\delta}(t)=-\beta(t)$. In order to proof that $y$ satisfy (\ref{mKdV-SG}), we need to proof that
\begin{equation}\label{proof1}
\biggl({\partial\over\partial x}-y_x\biggr){\cal L}_{n}\biggl[x;-y_{xx}-{1\over2}y_x^2\biggr]=-\biggl({\partial\over\partial x}+y_x\biggr){\cal L}_{n}\biggl[x;y_{xx}-{1\over2}y_x^2\biggr]
\end{equation}

Observe we can rewrite equation (\ref{proof1}) as
\begin{equation}\label{cond}
{\partial\over\partial x}{\cal L}_{n}^{(e)}\biggl[x;y_{xx}-{1\over2}y_x^2\biggr]+y_x{\cal L}_{n}^{(o)}\biggl[x;y_{xx}-{1\over2}y_x^2\biggr]=0
\end{equation}

By deriving twice the above equation, we have
\begin{equation}\label{auxiliar}
{\partial^3\over\partial x^3}{\cal L}_{n}^{(e)}\biggl[x;y_{xx}-{1\over2}y_x^2\biggr]+y_{xxx}{\cal L}_{n}^{(o)}\biggl[x;y_{xx}-{1\over2}y_x^2\biggr]+2y_{xx}{\partial\over\partial x}{\cal L}_{n}^{(o)}\biggl[x;y_{xx}-{1\over2}y_x^2\biggr]+y_{x}{\partial^2\over\partial x^2}{\cal L}_{n}^{(o)}\biggl[x;y_{xx}-{1\over2}y_x^2\biggr]=0
\end{equation}

Expression (\ref{cond}) is verified by induction. It is ease to verify that (\ref{cond}) is true for $n=0$ and $n=1$, i.e.
\[
{\cal L}_{0}\biggl[x;y_{xx}-{1\over2}y_x^2\biggr]=1  \hspace{0.5 cm}  \Rightarrow   \hspace{0.5 cm} {\cal L}_{0}^{(e)}\biggl[x;y_{xx}-{1\over2}y_x^2\biggr]=1   ,   \hspace{0.5 cm} {\cal L}_{0}^{(o)}\biggl[x;y_{xx}-{1\over2}y_x^2\biggr]=0
\]
\[
{\cal L}_{1}\biggl[x;y_{xx}-{1\over2}y_x^2\biggr]=y_{xx}-{1\over2}y_x^2  \hspace{0.5 cm}  \Rightarrow   \hspace{0.5 cm} {\cal L}_{1}^{(e)}\biggl[x;y_{xx}-{1\over2}y_x^2\biggr]=-{1\over2}y_x^2   ,   \hspace{0.5 cm} {\cal L}_{1}^{(o)}\biggl[x;y_{xx}-{1\over2}y_x^2\biggr]=y_{xx} .
\]

Now, let us assume that expression (\ref{cond}) is true for $n=k-1$. Using (\ref{Lenard e}) and (\ref{Lenard o}) in expression  (\ref{cond}) with $n=k$ yields
\begin{eqnarray*}
& &  \biggl({\partial^3\over\partial x^3}-y_x^2{\partial\over\partial x}-{1\over2}(y_x^2)_x\biggr){\cal L}_{k-1}^{(e)}\biggl[x;y_{xx}-{1\over2}y_x^2\biggr]+\biggl(2y_{xx}{\partial\over\partial x}+y_{xxx}\biggr){\cal L}_{k-1}^{(o)}\biggl[x;y_{xx}-{1\over2}y_x^2\biggr]+ \\
& &  y_x\biggl\{\biggl({\partial^2\over\partial x^2}-y_x^2\biggr){\cal L}_{k-1}^{(o)}\biggl[x;y_{xx}-{1\over2}y_x^2\biggr]+y_{xx}{\cal L}_{k-1}^{(e)}\biggl[x;y_{xx}-{1\over2}y_x^2\biggr]+ \\
& & \int^x y_{xx}\biggl({\partial\over\partial x}{\cal L}_{k-1}^{(e)}\biggl[x;y_{xx}-{1\over2}y_x^2\biggr]+y_x{\cal L}_{k-1}^{(o)}\biggl[x;y_{xx}-{1\over2}y_x^2\biggr]\biggr)dx \biggr\}=0
\end{eqnarray*}

Therefore, by using (\ref{cond}) and  (\ref{auxiliar}) with $n=k-1$, the above equation is verified.

\end{proof}

\begin{theorem}\label{escala}
Let $y=y(x,t)$ be the solution of $E_n(y;\epsilon_0(t),\beta(t),\delta(t))$, then the transformation
\[
\tilde{y}=y(x,\tilde{t})+\ln f(\tilde{t}) ,  \hspace{1 cm}  \tilde{t}=\int^t{dt'\over \epsilon_0(t')}
\]
gives the solution for $E_n\biggl(\tilde{y};1,{\beta(t)f(\tilde{t})\over\epsilon_0(t)},{\delta(t)\over f(\tilde{t})\epsilon_0(t)}\biggr)$, with $t=t(\tilde{t})$, provided that $ \epsilon_0(t) \neq 0$ and $f(\tilde{t}) \neq 0$.

\end{theorem}

\begin{proof}
It is a direct substitution.
\end{proof}

Using theorem \ref{escala} with $f(\tilde{t})=\pm i\sqrt{\delta(t)\over\beta(t)}$, we can reduce the hierarchy (\ref{mKdV-SG Ap}), with $\epsilon_0(t)$, $\beta(t)$ and $\delta(t)$ non-null, to the usual mKdV-sinh-Gordon hierarchy
\begin{equation}\label{mKdV-SG usual}
{\partial\over\partial x}\biggl({\partial\over\partial x}+\tilde{y}_x\biggr){\cal L}_{n}\biggl[x;\tilde{y}_{xx}-{1\over2}\tilde{y}_x^2\biggr]+\tilde{y}_{x\tilde{t}}+\eta(\tilde{t})\sinh \tilde{y}=0 , \hspace{1 cm}  \eta(\tilde{t})=\pm{2 i\sqrt{\delta(t)\beta(t)}\over\epsilon_0(t) }
\end{equation}

The  hierarchy (\ref{mKdV-SG Ap}) with  $\delta(t)=0$ is the mixed mKdV-Liouville hierarchy and, by theorem \ref{escala},  can be reduced to  
\begin{equation}\label{mKdV-L}
{\partial\over\partial x}\biggl({\partial\over\partial x}+\tilde{y}_x\biggr){\cal L}_{n}\biggl[x;\tilde{y}_{xx}-{1\over2}\tilde{y}_x^2\biggr]+\tilde{y}_{x\tilde{t}}+\tilde{\beta}(\tilde{t}) e^{\tilde{y}}=0 , \hspace{1 cm}  \tilde{\beta}(\tilde{t})={\beta(t)f(\tilde{t})\over\epsilon_0(t) }
\end{equation}

Observe that the case $\beta(t)=0$ and $\delta(t) \neq 0$ can be mapped on hierarchy (\ref{mKdV-L}) by theorem \ref{invertion}.




\bibliographystyle{model1a-num-names}
\bibliography{<your-bib-database>}



\end{document}